\documentclass[prd,aps,eqsecnum,nofootinbib,amssymb,superscriptaddress,12pt]{revtex4}
\usepackage{amsmath}
\usepackage{amsthm}
\usepackage{mathtools}
\usepackage{graphicx}
\usepackage{amsfonts}

\usepackage{color, soul}

\usepackage{etoolbox}
\usepackage{lipsum}
\usepackage[colorlinks = true]{hyperref}

\hypersetup{
    pdftitle = {Title},
    pdfauthor = {Authors}
}

\usepackage{xcolor}
\definecolor{darkred}  {rgb}{0.5,0,0}
\definecolor{darkblue} {rgb}{0,0,0.5}
\definecolor{darkgreen}{rgb}{0,0.5,0}
\hypersetup{
    urlcolor   = blue,         
    linkcolor  = darkblue,     
    citecolor  = darkgreen,    
    filecolor  = darkred       
}

\newcommand{\old}{\color{black}} 

 \newtheorem{Def}{Definition}
 \newtheorem{Thm}{Theorem}
 \newtheorem{Lem}{Lemma}

\begin{document}
\title{Entropy, gap and a multi-parameter deformation of the Fredkin spin chain}
\author{Zhao Zhang and Israel Klich\\ Department of Physics, University of Virginia, Charlottesville, 22904, VA}

\begin{abstract}
We introduce a multi-parameter deformation of the Fredkin spin $1/2$ chain whose ground state is a weighted superposition of Dyck paths, depending on a set of parameters $t_i$ along the chain. The parameters are introduced in such a way to maintain the system frustration-free while allowing to explore a range of possible phases. In the case where the parameters are uniform, and a color degree of freedom is added we establish a phase diagram with a transition between an area law and a volume low. The volume entropy obtained for half a chain is $n \log s$ where $n$ is the half-chain length and  $s$ is the number of colors. Next, we prove an upper bound on the spectral gap of the $t>1, s>1$ phase, scaling as $\Delta=O((4s)^nt^{-n^2/2})$, similar to a recent a result about the deformed Motzkin model, albeit derived in a different way. Finally, using an additional variational argument we prove an exponential lower bound on the gap of the model for $t>1, s=1$, which provides an example of a system with bounded entanglement entropy and a vanishing spectral gap. \old
\end{abstract}

\maketitle    

\section{Introduction}

Entanglement is one of the central quantum phenomena that distinguish quantum systems from their classical counterparts. It has profound implications in many different contexts of modern physics and has both driven theoretical discussions of the foundations of quantum mechanics and motivated promising applications in quantum information and quantum computation. In quantum many-body physics, the study of entanglement is focused on the survival, or lack thereof, of entanglement between individual particles when a large number of particles organize themselves with interactions in a condensed matter system \cite{laflorencie2016quantum}. Some of the characteristics of the correlations generated this way are quantified by the concept of `entanglement entropy' and it's scaling with the system size. 

Scaling of entanglement entropy is often closely related to the spectral gap of a system, although an exact relation is still somewhat elusive, especially in higher dimensions. In a gapped system, correlations are short-ranged, and entanglement entropy is expected to obey an ``area law'', which says that the entanglement entropy of a region scales with the area of the boundary of the region as opposed to its volume. In one dimension the entanglement entropy of a subsystem is bounded by a constant regardless of its length, as shown in  \cite{hastings2007area, arad2013area}. For gapless systems, (1+1)-dimensional conformal field theory studies find logarithmic scaling \cite{holzhey1994high,calabrese2009entanglement, Jin2004}. Fermi liquids exhibit a logarithmic violation of area law behavior scaling in any dimension \cite{wolf2006violation,gioev2006entanglement}, and require control of multidimensional dimensional gineralizations of the Szego limit theorems as described by Widom's conjecture \cite{gioev2006entanglement,helling2010special,leschke2014scaling}.

Recently, efforts have been made to find one-dimensional spin systems with more severe violations of the area law \cite{irani2010ground,gottesman2010entanglement,vitagliano2010volume,ramirez2014conformal}. These exotic scalings are achieved at the price of either introducing a high-dimensional local Hilbert space or sacrificing translational invariance. Based on previous work by Bravyi et al \cite{bravyi2012criticality}, Movassagh and Shor \cite{movassagh2016supercritical} first introduced a highly-entangled (power-law scaling) integer spin-$s$ chain ($s>1$) that enjoys several physically appealing features, including relatively small dimensionality of the local Hilbert space, short-range interaction and translational invariance. Furthermore, the model Hamiltonian is frustration free: it can be written as a sum of local projectors sharing a unique simultaneous ground state. Inspired by this work, in \cite{zhang2016quantum} we have used the idea to describe a class of Hamiltonians with a tunable parameter $t$ that exhibits a novel quantum phase transition. 

The $t>1$ phase, in particular, features an enhanced scaling of entanglement entropy to a full extensive (volume) scaling. 
The phase exhibits additional unusual properties. For example, the spectral gap has been shown to decrease exceptionally fast as a function of system size $L=2n$. An initial estimate of exponential decay has been improved by Levine and Movassagh \cite{levine2016} to an unusually fast decay as $\exp(-n^2/3)$.
Recently, numerical studies on the nature of the gap in $t<1$ phase in the deformed Motzkin model have been carried out suggesting a type of topological order in the Haldane phase  \cite{barbiero2017haldane}.

Another realization of the new bound to extensive phase transition is based on the Fredkin model introduced by Salberger and Korepin \cite{salberger2016fredkin}. The Fredkin model utilizes half-integer spin chains, with next nearest neighbor interactions based on the so-called Fredkin gate \cite{dell2016violation, salberger2016fredkin}. The deformed Fredkin model is presented in \cite{salberger2016deformed} and exhibits a similar quantum phase transition into an extensively entangled state.
In particular,  a proof of the entropy scaling is given, utilizing known property of Dyck paths and their relation to Young Tableaux enumeration. 

In this paper we expand the results of \cite{salberger2016deformed} by addressing several important issues. Our main results are the following:

1) The underlying conditions for the validity of the frustration free deformation in \cite{salberger2016deformed} are clarified. In doing so we uncover a large class of deformations, that maintain a frustration free ground state.

2) We give an alternative, elementary and fully self-contained version of the proof of entropy scaling that might yield further insights of the underlying physics.
 
3) We prove an upper bound on the scaling of the spectral gap using a variational wave function that explicitly manifests the super-exponential decay of the gap giving a physical intuition behind the scaling established in \cite{levine2016} for the deformed Motzkin model.

4) Recent numerical evidence \cite{barbiero2017haldane} suggests that, surprisingly, the $t>1$ phase in the colorless Fredkin model is gapless. This behavior is somewhat unusual for this type of transition as the entanglement obeys an area law for $s=1$ and $t>1$. 
We explain this behavior by establishing a variational excited state that has vanishing energy $t>1, s=1$ in the thermodynamic limit. \old

The paper is structured as follows. In Section \ref{HGS}, we briefly review the constituting ingredients Dyck walk and Fredkin interaction of the Hamiltonian and ground state of Fredkin spin chain and show how the deformation parameter $t$ fit in without compromising its frustration free feature of the Hamiltonian and the uniqueness of its ground state. In Section \ref{EntEnt}, we illuminate the mechanism behind the linear scaling of entanglement entropy of Fredkin spin chain by introducing the non-commutative `height' and `shift' operators, without resorting to any mathematical knowledge beyond analysis. In Section \ref{SpGp}, we exploit features of the Hamiltonian and its ground state to construct low energy excitation states for the $s>1$ and $s=1$ models that have exponentially small gaps in the $t>1$ region of the phase diagram. Finally, in Section \ref{Discussion}, we point out several possible future directions.

\section{Hamiltonian and ground state}\label{HGS}

We start with a brief review of Dyck walks, which span the ground state subspace of the Hilbert space of the Fredkin spin chain Hamiltonian introduced in \cite{dell2016violation,salberger2016fredkin}. 
\begin{Def} A Dyck walk (or path) on $2n$ steps is any path from $(0, 0)$ to $(0, 2n)$ with steps $(1, 1)$ and $(1, -1)$ that never passes below the $x$-axis. \end{Def}
A Dyck walk can be mapped to a spin configuration of a spin-$1/2$ chain $|\sigma_1 \sigma_2 \dots \sigma_{2n}\rangle$ with $\sigma_k = +1/2$ or $-1/2$ for a $(1, 1)$ or $(1, -1)\ k$th step respectively. When each step is assigned a color $c_k$ from a palette of $s$ colors, $|\uparrow^c_k\rangle$ ($|\downarrow^c_k\rangle$) corresponds to a state $\sigma_k = c_k/2$ (resp. $\sigma_k = -c_k/2$). And a colored Dyck walk can, therefore, represent the spin configuration of any half-integer spin chain. In the original Fredkin spin chain model, the ground state is a uniform superposition of colored Dyck walks. Here following \cite{zhang2016quantum}, we present a class of Hamiltonians that feature higher probability amplitudes in the superposition for paths with greater heights. Our Hamiltonians have a reduced contribution to entanglement from fluctuation in path shape, which are greatly reduced, however we have an enhanced contribution to entanglement from color degrees of freedom, which, in turn, is responsible for the volume scaling of entanglement.

We introduce a parameter $t$ that deforms the Fredkin Hamiltonian of \cite{dell2016violation,salberger2016fredkin} while remaining frustration free. The Hamiltonian is given by:
\begin{equation} H(s,t) = H_F(s,t) + H_X(s) + H_{\partial}(s), \label{Ham} 
\end{equation} 
where 
\begin{eqnarray}
H_{F}(s,t)=\sum_{j=2}^{2n-1}\sum_{c_1,c_2,c_3=1}^{s}\big(|\phi^{{c_{1},c_{2},c_{3}}}_{j,A}\rangle\langle\phi^{{c_{1},c_{2},c_{3}}}_{j,A}|+|\phi^{{c_{1},c_{2},c_{3}}}_{j,B}\rangle\langle\phi^{{c_{1},c_{2},c_{3}}}_{j,B}|\big) \label{HF}
\end{eqnarray}
\begin{eqnarray}
H_X(s) =& \sum_{j=1}^{2n-1}  [\sum_{c_1 \neq c_1} |\uparrow_j^{c_1}\downarrow_{j+1}^{c_2}\rangle \langle\uparrow_j^{c_1}\downarrow_{j+1}^{c_2}|\\ & + \frac{1}{2} \sum_{c_1,c2=1}^s (|\uparrow_j^{c_1}\downarrow_{j+1}^{c_1}\rangle - |\uparrow_j^{c_2}\downarrow_{j+1}^{c_2}\rangle)(\langle \uparrow_j^{c_1}\downarrow_{j+1}^{c_1}| - \langle\uparrow_j^{c_2}\downarrow_{j+1}^{c_2}|)],
\end{eqnarray}
and
\begin{eqnarray}
H_{\partial}(s) =& \sum_{c=1}^s (|\downarrow_1^c\rangle\langle\downarrow_1^c| + |\uparrow_{2n}^c\rangle\langle\uparrow_{2n}^c|).
\end{eqnarray}
The projectors in $H_F$ are defined using:
\begin{eqnarray}
|\phi^{{c_{1},c_{2},c_{3}}}_{j,A}\rangle=\frac{1}{\sqrt{1+|t_{A,j}|^2}}(\left|\uparrow_{j}^{c_1}\uparrow_{j+1}^{c_2}\downarrow_{j+2}^{c_3}\right\rangle-t_{A,j}\left|\uparrow_{j}^{c_2}\downarrow_{j+1}^{c_3}\uparrow_{j+2}^{c_1}\right\rangle)
\end{eqnarray}
\begin{eqnarray}
|\phi^{{c_{1},c_{2},c_{3}}}_{j,B}\rangle=\frac{1}{\sqrt{1+|t_{B,j}|^2}}(\left|\uparrow_{j}^{c_1}\downarrow_{j+1}^{c_2}\downarrow_{j+2}^{c_3}\right\rangle-t_{B,j}\left|\downarrow_{j}^{c_3}\uparrow_{j+1}^{c_1}\downarrow_{j+2}^{c_2}\right\rangle)\label{phiminus}\label{phiB}
\end{eqnarray}
with the condition that $t_j^B=t_{j-1}^A$.

The Fredkin gate projectors in $H_F$ allows a pair of $\uparrow\downarrow$ neighboring spins (with the same color enforced by the first term in $H_X$) to move freely around its left or right third neighbor and still appear in the ground state superposition, but now with a different probability amplitude. The second term in $H_X$ ensures that otherwise identical Dyck paths with different coloring have the same weight. And the boundary term $H_{\partial}$ (together with the Fredkin projectors) penalizes paths that go below $0$ at any point along the chain. Notice that analogous to \cite{zhang2016quantum}, the simplest choice is a parameter $t=t_A=t_B$ being the same in the two projectors of $H_F$ is the one employed in \cite{salberger2016deformed,udagawa2017finite}, but is only a subset of the parameter space that leaves the Hamiltonian frustration free. More generally, we introduce parameters $t^A$ and $t^B$, for the two projectors in $H_{F}$. Then any set of $\{t_j^A, t_j^B\}$ that satisfies the condition $t_j^A = t_{j+1}^B$ for all $j$'s would guarantee the Hamiltonian to be frustration free \footnote{Note that there is no parallel condition $t_j^B = t_{j+1}^A$.}. The point is illustrated in Fig. \ref{ffcondition}. In particular, we may specify a Hamiltonian with a frustration-free ground state by picking any set of the $t_j^A$ parameters.
\begin{figure} \centering \includegraphics[width=0.4\textwidth]{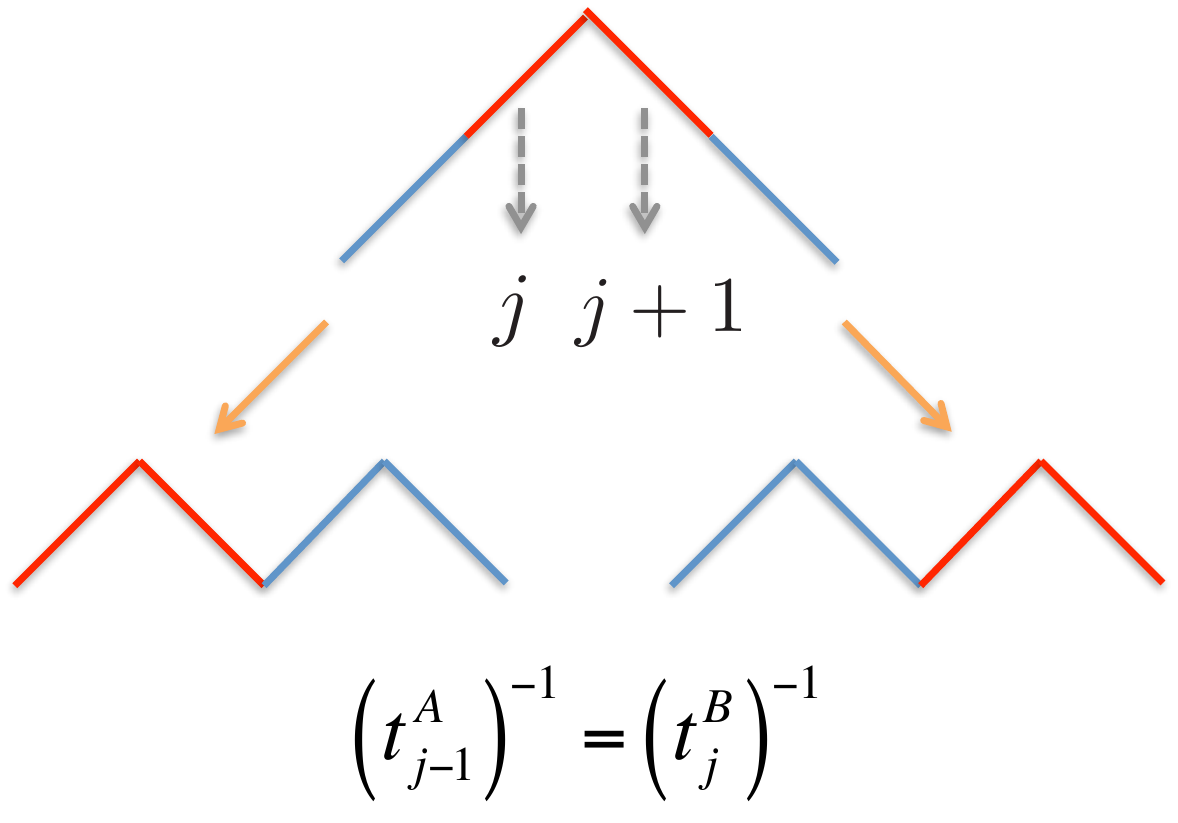} \caption{Different ways of ``flattening'' a hill must have the same amplitude. Note that the colors have been interchanged during the procedure, this, however, is not an obstruction as all colorings appear with the same amplitude.} \label{ffcondition} \end{figure}

We now want to characterize the ground state of the system. First, let us denote $h(l)$ to be the height of the Dyck path after step $l$, that is, for a spin configuration $|w\rangle$ describing a Dyck path,  \begin{eqnarray}
\sum_{c = 1}^s\sum_{j=1}^{l} \sigma_{j,c}^{z} |w\rangle=h(l)|w\rangle,
\end{eqnarray}
where $\sigma_{j,c}^z$ is the Pauli matrix giving $\pm 1$ if  spin $j$ is in state $\uparrow_j^c$ or $\downarrow_j^c$, respectively.
The height function is illustrated for a generic Dyck path in Fig. \ref{Heights}.
\begin{figure} \centering \includegraphics[width=0.8\textwidth]{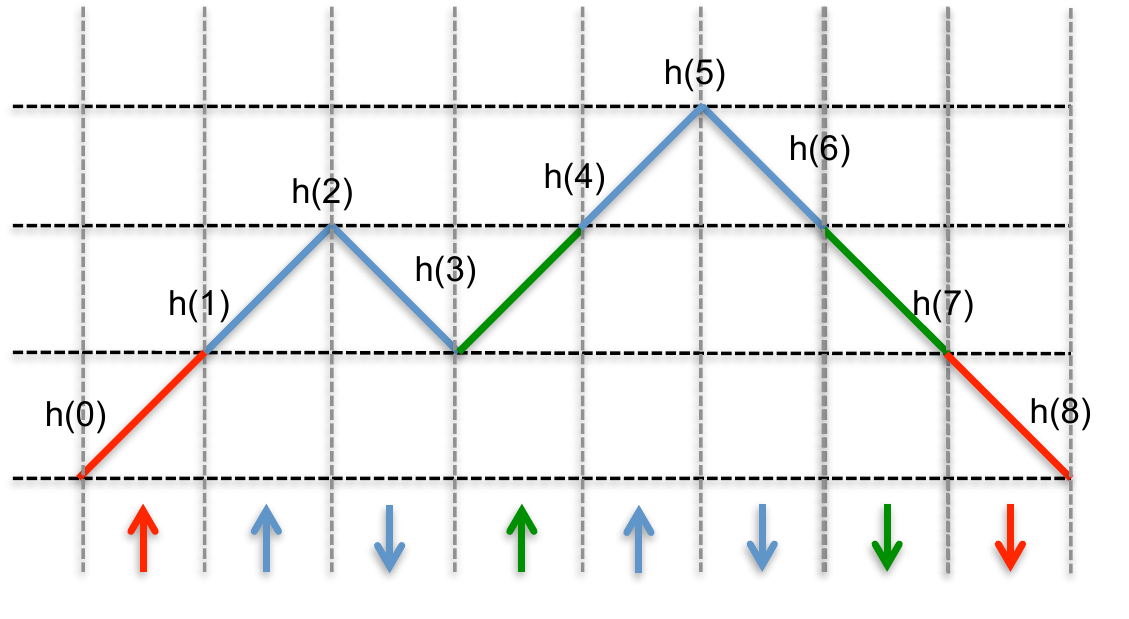} \caption{A spin configuration corresponding to a colored Dyck path and the corresponding height function $h(l)$.} \label{Heights} \end{figure}
To find the relative amplitude of this spin configuration as compared with the lowest possible spin configuration, we use successively the Fredkin moves to $|\uparrow \uparrow \downarrow\rangle \longrightarrow t_A|\uparrow \downarrow\uparrow \rangle $ to "flatten" the hill. The process is described in Fig. \ref{flattening}.
\begin{figure} \centering \includegraphics[width=1\textwidth]{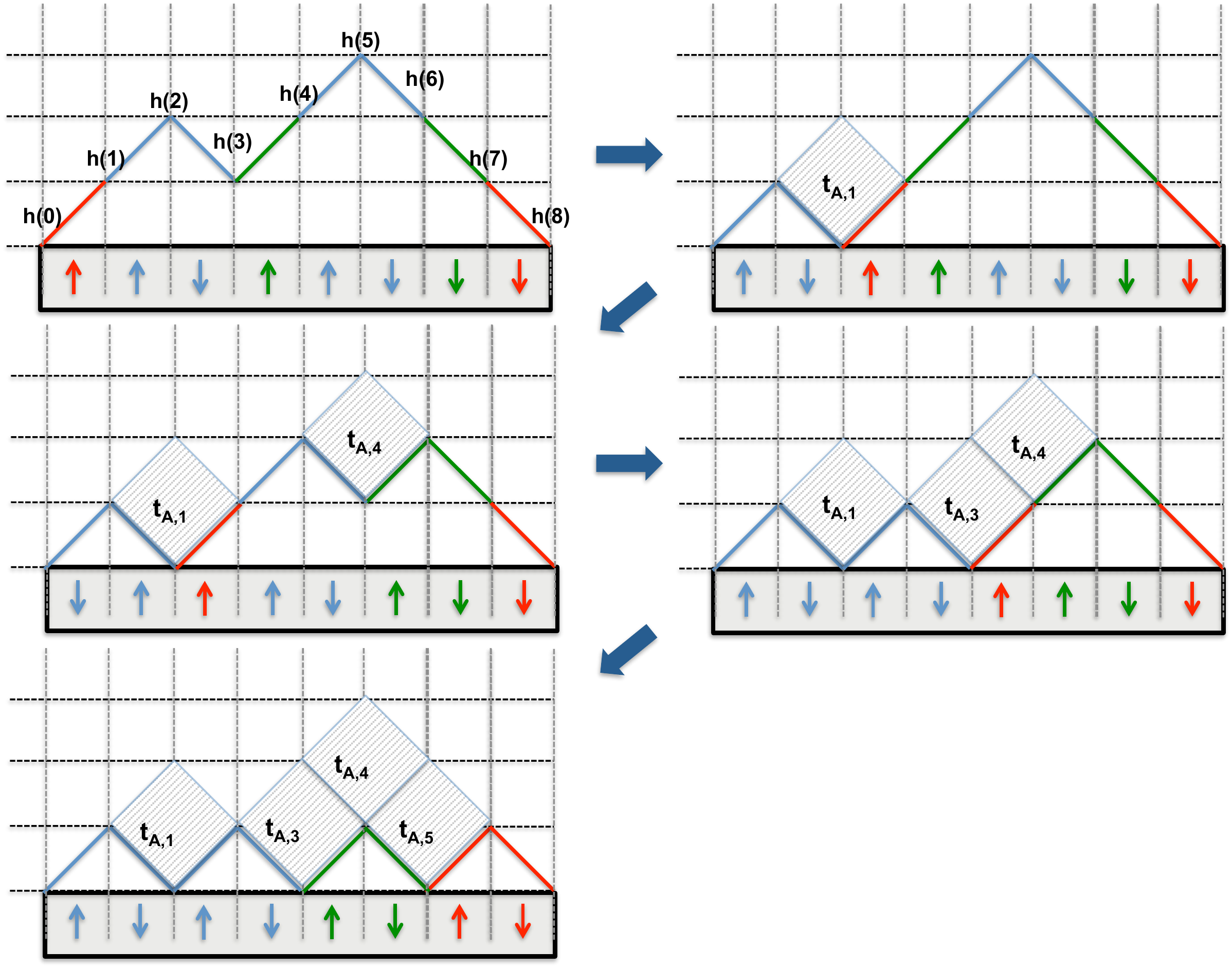} \caption{The ``flattening'' of a hill and it's amplitude. Starting from the left, we reduce the first by using the Fredkin move $|\uparrow \uparrow \downarrow\rangle \longrightarrow t_A|\uparrow \downarrow\uparrow \rangle $. This process is repeated for each peak until the lowest height hill is achieved.} \label{flattening} \end{figure}
In this way, the weight of each Dyck paths is related to the  weight of the  lowest height, $|\uparrow\downarrow\uparrow\downarrow...\uparrow\downarrow\rangle$ path. Note that we have suppressed the color index in this treatment  since, as mentioned above, in the ground state superposition all admissible colorings should  appear with the same amplitude.

The amplitude of a given Dyck path in the ground state of the model is thus given by counting the number of "diamonds" associated with each  $t_j^A$, and can be written in the form:
\begin{equation} |GS\rangle = \frac{1}{\mathcal{N}}\sum_{\substack{w \in \{s-colored\\ Dyck\ walks\}}} e^{\sum_{l=2}^{2n-2}[\frac{h(l)}{2}]\log (t_{l-1}^A)}|w\rangle, \label{eq: gs}\end{equation} 
where $[x]$ is the integer part of $x$ and $\mathcal{N}$ is a normalization factor.
In the case where $t_j^A\equiv t$, the ground state can be simply related to the area under the path as: 
\begin{equation} |GS\rangle = \frac{1}{\mathcal{N}'}\sum_{\substack{w \in \{s-colored\\ Dyck\ walks\}}} t^{\frac{1}{2}\mathcal{A}(w)}|w\rangle. \label{eq: gs}\end{equation} 

\section{Entanglement entropy} \label{EntEnt}

In this section, we employ the simplest choice of parameters which is translational invariant $t_A=t_B=t$ everywhere. 
When $t=1$, the entanglement entropy of the ground state scales as $\log n$ for $s=1$ and as $\sqrt{n}$ for $s>1$ \cite{dell2016violation,salberger2016fredkin}. The reason our deformation with the extra parameter can further increase the scaling of entropy is because when a spin is moved around its neighboring $\uparrow\downarrow$ pair, it is separated from its own partner paired in the same color, which is the first unpaired down spin to its right (or up spin to its left). This way, when a pair of spins required to be in the same color are shifted to different subsystems of the chain, they become a source of entanglement entropy between the two subsystems. Tuning the parameter $t$ to favor higher paths in the ground state superposition will now enhance the more substantial contribution from those with more unpaired spins in one subsystem. To put this in a mathematical way, we decompose the ground state into tensor products of states in the left and right halves of the chain. 
\begin{equation} 
|GS\rangle = \sum_{m=0}^{n} \sqrt{p_{n,m}} \sum_{x\in \{\uparrow^1,\uparrow^2,\ldots,\uparrow^s\}^m} |\hat{C}_{0,m,x}\rangle_{1,\ldots,n} \otimes |\hat{C}_{m,0,\bar{x}}\rangle_{n+1,\ldots,2n}, 
\end{equation} 
where $|\hat{C}_{p,q,x}\rangle$  is a weighted superposition of spin configurations with $p$ excess $\downarrow$, $q$ excess $\uparrow$ and a particular coloring $x$ of the unmatched arrows, such that $\langle GS|(|\hat{C}_{0,m,x}\rangle_{1,\ldots,n} \otimes |\hat{C}_{m,0,\bar{x}}\rangle_{n+1,\ldots,2n}) \neq 0$, and $\bar{x}$ is the coloring in the second half of the chain that matches $x$. The decomposition gives the Schmidt number \begin{equation} p_{n,m}(s,t) = \frac{M_{n,m}^2(s,t)}{N_n(s,t)}, \end{equation} where \begin{align}M_{n,m}(s,t) &\equiv s^{\frac{n-m}{2}}\sum_{\substack{w\in \{1st\ half\ of\ Dyck\\ walks\ stopping\ at\ (n,m)\}}}t^{\mathcal{A}(w)},\\ N_n(s,t) &\equiv \sum_{m=0}^{n}s^mM_{n,m}^2(s,t). \end{align}  And the entanglement entropy of the half chain in the ground state is given by 
\begin{equation} 
S_n(s,t) = -\sum_{m=0}^ns^mp_{n,m}(s,t)\log p_{n,m}(s,t).
\end{equation}

To study the behavior of $M_{n,m}$ as a function of $m$, we observe that they satisfy the following recurrence relations, \begin{equation} \begin{aligned} M_{k+1,k+1} &= t^{k+\frac{1}{2}} M_{k,k},\\ M_{k+1,m} &= st^{m+\frac{1}{2}} M_{k,m+1} + t^{m-\frac{1}{2}} M_{k,m-1}, \quad \text{for}\ 0<m<k, \\ M_{k+1,0} &= st^{\frac{1}{2}} M_{k,1}. \end{aligned}\label{eq: recur}\end{equation} Notice that the $M_{n,m}$ is only non-vanishing for $m$'s of same parity as $n$. 

From these relations, we can see that for large enough $t$, $M_{n,m}$ will be monotonically increasing as we increase $m$ by increments of $2$. Paths with height in the middle scaling as $O(n)$ will contribute more to the entanglement entropy from the $s^{O(n)}$ possible colorings of unmatched spins. In particular, the half chain entanglement entropy will also scale linearly with system size $n$. In the next subsection, we give a rigorous proof that this is true in the thermodynamic, and that this critical phase of large entanglement spans the entire half line $t>1$.

\subsection{$t>1,s>1$ phase: Volume scaling of entropy.}

In this section we repeat the steps taken in \cite{zhang2016quantum}, to prove volume scaling for weighted Motzkin walks, with a few modifications. For arbitrary $t>1$, the non-zero entries of $M_{n,m}$ are not necessarily monotonic in terms of $m$, but we can still show that for a given $n$, $M_{n,m}$ reaches its maximum at some $m=m^*$, within a finite distance away from $m=n$ independent of the system size $n$ itself. This is not obvious in the step-by-step recurrence relations, but becomes clear as we take into account the accumulated effect of the evolution of the coefficients $M_{n,m}$ with respect to $n$. To see this, we summarize \eqref{eq: recur} in the following operator formalism. 

As in \cite{zhang16Quantum}, we represent the distributions of $M_{k,m}$ as components of the state at `time' $k$ during the `evolution' in a basis spanned by $|m\rangle$, $m=0,1,2,\ldots$.
\begin{equation} |\mathcal{M}_k\rangle = \sum_{m=0}^\infty M_{k,m} |m\rangle, \qquad M_{k,m} = 0\ \text{if}\ m>k.\end{equation} 
We we define `shift' and `height' operators to describe the `evolution' of the the states $|\mathcal{M}_k\rangle$ as 
\begin{align} \mathcal{S} |m\rangle &= |m-1\rangle,  \mathcal{S} |0\rangle= 0; \\ \mathcal{H} |m\rangle &= m |m\rangle. \end{align}
One can check that the recurrence relations \eqref{eq: recur} translate to
\begin{align*}&\langle m | \mathcal{M}_{k+1}\rangle = M_{k+1,m} = s t^{m + \frac{1}{2}} \langle m+1|\mathcal{M}_k\rangle + t^{m - \frac{1}{2}} \langle m-1|\mathcal{M}_k\rangle \\ = &\langle m | s t^{\mathcal{H} + \frac{1}{2}} \mathcal{S} + t^{\mathcal{H} - \frac{1}{2}} \mathcal{S}^\dagger |\mathcal{M}_k\rangle, \end{align*} 
which gives us:
\begin{equation} |\mathcal{M}_{k+1}\rangle = t^{\mathcal{H}} (s\sqrt{t}\mathcal{S} + \frac{1}{\sqrt{t}}\mathcal{S}^\dagger)|\mathcal{M}_k\rangle. \end{equation} 

Using the commutation relation
\begin{equation} t^{k\mathcal{H}} (s\sqrt{t}\mathcal{S}  + \frac{1}{\sqrt{t}}\mathcal{S}^\dagger) = (st^{-(k - \frac{1}{2})} \mathcal{S}  + t^{k - \frac{1}{2}} \mathcal{S}^\dagger) t^{k\mathcal{H}}, \end{equation}
we keep moving the $t^{k\mathcal{H}}$ operators all the way to the right until it disappears when acting on 
$|0\rangle$ we obtain:
\begin{align} |\mathcal{M}_n\rangle &= [t^{\mathcal{H}} (s\sqrt{t}\mathcal{S} + \frac{1}{\sqrt{t}}\mathcal{S}^\dagger)]^n |\mathcal{M}_0\rangle = \vec{\mathcal{K}}\prod_{k=1}^n  (st^{-(k - \frac{1}{2})} \mathcal{S} + t^{k - \frac{1}{2}} \mathcal{S}^\dagger) |0\rangle. \label{korder}\end{align}
Here $\vec{\mathcal{K}}$ denotes ordering the multiplications in the product such that factors with greater $k$ value are on the right. For $t>1$ the factors in the product above are dominated by the $\mathcal{S}^\dagger$ term for large $k$. In other words, at some point during the evolution, the distribution of $M_{m}$ starts shifting at velocity $1$ to the right along the $m$ axis without much spreading. For a larger $t$, this happens shortly after the evolution starts, while for smaller values of $t$, it takes longer to reach this stable propagation. In any case, as we show below, the maximum of $M_{n,m}$ is a within finite distance away from $m=n$.

\begin{Lem} Let $m^*$ be such that $\sup_{m} M_{n,m} = M_{n,m^*}$, then $\exists N_{0} < n$, such that when $t > 1$, $m^* \in [n - 2N_{0}, n]$. \label{Lemma:mstar}\end{Lem}
\begin{proof} Let 
\begin{equation} |\mathcal{M'}_n\rangle =\vec{\mathcal{K}}\prod_{k={N_{0}}+1}^{n} (st^{-(k - \frac{1}{2})} \mathcal{S} + t^{k - \frac{1}{2}} \mathcal{S}^\dagger) |0\rangle. \end{equation}

Note that
\begin{equation} t^{-(k-\frac{1}{2})}\| st^{-(k-\frac{1}{2})} \mathcal{S} \|_1 \le st^{-2(k-\frac{1}{2})} \equiv c_k,\end{equation} 
so that:
\begin{eqnarray*} &
\| st^{-(k - \frac{1}{2})}\mathcal{S} + t^{k - \frac{1}{2}} \mathcal{S}^\dagger \|_1<t^{k - \frac{1}{2}}(1+c_{k})
\end{eqnarray*} 
we thus have 
\begin{eqnarray*} &  \| t^{-\sum_{k=N_{0}+1}^{n}(k - \frac{1}{2})} |\mathcal{M'}_n\rangle - |n-{N_{0}}\rangle \|_1 \le \prod_{k=N_{0}+1}^{n}(c_k+1) - 1 < e^{\sum_{k=N_{0}+1}^\infty c_k} -1 \\ & = e^{\frac{st}{t^2-1}t^{-2N_{0}}}-1 \equiv f(s,t)^{t^{-2N_{0}}} - 1.
\end{eqnarray*} 
The first inequality on the left follows from noting that $ |n-{N_{0}}\rangle$ appears in $t^{-\sum_{k=N_{0}+1}^{n}(k - \frac{1}{2})} |\mathcal{M'}_n\rangle $ with coefficient 1, and is exactly canceled. We have also used that $x+1\le e^x$.
Next,
\begin{align*} &\| t^{-\sum_{k=1}^{n}(k - \frac{1}{2})} |\mathcal{M}_n\rangle - \vec{\mathcal{K}}\prod_{k=1}^{N_{0}}(st^{-2(k-\frac{1}{2})}\mathcal{S} + \mathcal{S}^\dagger)|n-N_{0}\rangle \|_1 \\ \le &\| \vec{\mathcal{K}}\prod_{k=1}^{N_{0}}(st^{-2(k-\frac{1}{2})}\mathcal{S} + \mathcal{S}^\dagger) \|_1 \| t^{-\sum_{k=N_{0}+1}^{n}(k - \frac{1}{2})} |\mathcal{M'}_n\rangle - |n-N_{0}\rangle \|_1\\ < &(f(s,t)^{t^{-N_{0}}} - 1) \prod_{k=1}^{N_{0}}(1+c_k) < (f(s,t)^{t^{-N_{0}}} - 1)e^{\sum_{k=1}^{N_{0}}}c_k  < (f(s,t)^{t^{-N_{0}}} - 1)f(s,t). \end{align*} Let 
\begin{equation} M'_{n,m} = \langle m|\vec{\mathcal{K}}\prod_{k=1}^{N_{0}}(st^{-2(k-\frac{1}{2})}\mathcal{S} + \mathcal{S}^\dagger)|n-N_{0}\rangle, \end{equation} then clearly $M'_{n,m} = 0$ for $m < n - 2{N_{0}}$. 
If we choose 
\begin{equation} N_{0}= \begin{cases} 0 &  f(s,t) < \frac{1+ \sqrt{5}}{2}, \\ -\frac{\log \frac{\log (f^{-1}(s,t) + 1)}{\log f(s,t)}}{\log t}, &\text{otherwise},\end{cases} \end{equation} then 
\begin{equation} \| t^{-\frac{n^2}{2}} |\mathcal{M}_n\rangle - \sum_{m = n - 2 N_{0}}^n M'_{n,m} |m\rangle \|_1 < 1 = M'_{n,n} \le \sup_{m} M'_{n,m}. \end{equation} Therefore $\exists m^* \in [n-2N_{0},n]$, such that $M_{n,m^*} \geq M_{n,m}$ for all $m$. \end{proof}

This allows us to prove the linear scaling of the entanglement entropy.

\begin{Thm} In the state \eqref{eq: gs}, when $t>1$, the entanglement entropy of half of the chain is bounded from below by $S_n > n \log s + C(s,t)$, where $C(s,t)$ is an $n$ independent constant.\end{Thm}

\begin{proof} We separate a linear term from $S_{n}$ as follows (below we supress the $n$ index in $M_{n,m}$):  \begin{eqnarray} \nonumber &S_n = \sum_{m=0}^{n}s^m p_m\log s^m - \sum_{m=0}^n s^m p_m\log(s^m p_m)> \sum_{m=0}^n s^m p_m m\log s\\ &= \sum_{l=0}^n s^{n-l}p_{n-l}(n - l)\log s = n\log s - \log s \sum_{l=0}^n\frac{s^{n-l}M_{n-l}^2}{{\sum_{m'=0}^ns^{m'}M_{m'}^2}}l \label{lower entropy bound}
 \end{eqnarray} 
 Taking $m^{*}$ such that $\sup_{m} M_{n,m} = M_{n,m^*}$ and using lemma \ref{Lemma:mstar}, we see that 
 \begin{align*} 
   &   \sum_{l=0}^n\frac{s^{n-l}M_{m^*}^2}{{\sum_{m'=0}^ns^{m'}M_{m'}^2}}l <\sum_{l=0}^n\frac{s^{n-l}M_{m^*}^2}{{ s^{m^{*}}M_{m^{*}}^2}}l= s^{n - m^*} \sum_{l=0}^ns^{-l}l \\ &< s^{2N_{0}} \sum_{l=0}^ns^{-l}l < s^{2N_{0}} \sum_{l=0}^{\infty}s^{-l}l =    \frac{s^{2N_{0}+1}}{(s-1)^2} . \end{align*} Therefore, the remainder term on the right hand side of \eqref{lower entropy bound} is bounded. \end{proof}
One can see from the proof that the factor of $s^m$ is already enough to make the scaling of entropy linear, and all that is required for $p_m$ is that it doesn't destroy this exponential dependence on $m$. 

\subsection{$t<1$ and any $s$:  Bounded entanglement entropy.}   

Contrary to the case studied above, when $t<1$, we expect Dyck paths with smaller areas below to be exponentially preferred in the ground state superposition. But this time, for the entropy to reflect the predominance of lower path, where less mutual information between the two subsystems can be stored, the behavior of $p_m$ needs to not only be decreasing exponentially with $m$, but also fast enough to overcome the exponential increasing $s^m$ factor. Considering that, we define 
\begin{equation} \tilde{M}_{n,m}=s^{\frac{m}{2}}M_{n,m}~~,~~\tilde{p}_{n,m}=\frac{\tilde{M}_{n,m}^2}{\sum_{m=0}^n \tilde{M}_{n,m}^2}.\end{equation} Substitution into \eqref{eq: recur} gives the following relations,
\begin{equation} \begin{aligned} \tilde{M}_{k+1,k+1} &= \sqrt{s} t^{k + \frac{1}{2}}\tilde{M}_{k,k},\\
\tilde{M}_{k+1,m} &= \sqrt{s} (t^{m+\frac{1}{2}}\tilde{M}_{k,m+1} + t^{m-\frac{1}{2}}\tilde{M}_{k,m-1}), \quad \ 0<m<k,\\ \tilde{M}_{k+1,0} &= \sqrt{s} t^{\frac{1}{2}}\tilde{M}_{k,1}\label{tildeRec} \end{aligned} \end{equation}

To prove the entropy is bounded, we need the following lemmas.

\begin{Lem} \begin{equation} \tilde{M}_{n+2,m}^2 >  \tilde{M}_{n,m}^2. \end{equation}\label{Lem:MonotonicDenom} \end{Lem}

\begin{proof} From \eqref{korder}, we have \begin{align*} |\mathcal{M}_{n+2}\rangle &= \vec{\mathcal{K}}\prod_{k=1}^n  (st^{-(k - \frac{1}{2})} \mathcal{S} + t^{k - \frac{1}{2}} \mathcal{S}^\dagger) (st^{-(n + \frac{1}{2})} \mathcal{S} + t^{n + \frac{1}{2}} \mathcal{S}^\dagger)(st^{-(n+1 + \frac{1}{2})} \mathcal{S} + t^{n+1 + \frac{1}{2}} \mathcal{S}^\dagger)|0\rangle, \\ &= \vec{\mathcal{K}}\prod_{k=1}^n  (st^{-(k - \frac{1}{2})} \mathcal{S} + t^{k - \frac{1}{2}} \mathcal{S}^\dagger) [s^2 t^{-2(n+1)} \mathcal{S}^2 + s (t + \frac{1}{t}) + t^{2(n+1)} \mathcal{S}^{\dagger 2}] |0\rangle \\ &= s(t+\frac{1}{t})|\mathcal{M}_n\rangle + \vec{\mathcal{K}}\prod_{k=1}^n  (st^{-(k - \frac{1}{2})} \mathcal{S} + t^{k - \frac{1}{2}} \mathcal{S}^\dagger) [s^2 t^{-2(n+1)} \mathcal{S}^2 + t^{2(n+1)} \mathcal{S}^{\dagger 2}] |0\rangle.\end{align*} The last term on the RHS of the equation contains non-zero contributions for all states $|m\rangle$, with $m=0,1,\ldots, n+2$, and we have: \begin{align*} M_{n+2,m} &> M_{n,m},\\ \tilde{M}_{n+2,m} &> \tilde{M}_{n,m} \qquad \forall m\geq 0,n\geq 1. \end{align*} \end{proof}
Next we establish the following bound on $\tilde{p}_{n,m}$:
\begin{Lem} \begin{equation} \tilde{p}_{n,m} < 36\frac{s^2}{t^2} t^{4m}.\end{equation}  \label{Lemma:pnm bound} \end{Lem}

\begin{proof} By definition of $\tilde{p}_{n,m}$, and using the recursion relation \eqref{tildeRec} twice consequtively, 
\begin{align*} \tilde{p}_{n,m} &=  \frac{\{\sqrt{s} [t^{m+\frac{1}{2}}\sqrt{s} (t^{m+\frac{3}{2}}\tilde{M}_{n-2,m+2} + t^{m+\frac{1}{2}}\tilde{M}_{n-2,m}) +  t^{m-\frac{1}{2}}\sqrt{s}(t^{m-\frac{1}{2}}\tilde{M}_{n-2,m}+t^{m-\frac{3}{2}}\tilde{M}_{n-2,m-2})]\}^2}{\sum_{m=0}^{n} \tilde{M}_{n,m}^2} \\ &= \frac{s^2 t^{4m}[ t^{2}\tilde{M}_{n-2,m+2} + (t+\frac{1}{t})\tilde{M}_{n-2,m} +  t^{-2}\tilde{M}_{n-2,m-2}]^2}{\sum_{m=0}^{n} \tilde{M}_{n,m}^2} \\ &\le \frac{s^2 t^{4m}[3\max\{ t^{2}\tilde{M}_{n-2,m+2}, (t+\frac{1}{t})\tilde{M}_{n-2,m}, t^{-2}\tilde{M}_{n-2,m-2}\}]^2}{\sum_{m=0}^{n} \tilde{M}_{n,m}^2} \\ &\le 36\frac{s^2}{t^2} t^{4m} \frac{max\{ \tilde{M}_{n,m+2}^2, \tilde{M}_{n,m}^2, \tilde{M}_{n,m-2}^2\}}{\sum_{m=0}^n \tilde{M}_{n,m}^2} < 36\frac{s^2}{t^2} t^{4m}. 
\end{align*} Lemma \ref{Lem:MonotonicDenom} was used in the last line.\end{proof}

We now have the ingredients to prove the boundedness of entropy.

\begin{Thm} When $0<t<1, s\ge 1$, there exists a constant $C_0(s,t)$ independent of the system size $n$, that for any $n$, $S_n<C_0(s,t)$. \end{Thm}

\begin{proof} Using Lemma \ref{Lemma:pnm bound} we see that when \begin{equation} m > m_0 \equiv \Big[\frac{\log(\frac{1}{36e} {\frac{t^2}{s^2}})}{4 \log{t}}\Big] + 1, \end{equation} we have \begin{equation} \tilde{p}_{n,m} < 36\frac{s^2}{t^2} t^{4m}<{1\over e}. \end{equation} 
It is easy to check that the function $-x\log(x)$ is monotonically increasing when $x\in (0, \frac{1}{e})$, in other words, for $m>m_0$, 
\begin{eqnarray} \tilde{p}_{n,m}< 36\frac{s^2}{t^2} t^{4m}<{1\over e}~\Longrightarrow~ -\tilde{p}_{n,m}\log \tilde{p}_{n,m} < -36\frac{s^2}{t^2} t^{4m} \big(\log( {\frac{36s^2}{t^2}}) + 4m\log t\big). \end{eqnarray} 
Therefore \begin{align*} S_n &= -\sum_{m=0}^n \tilde{p}_{n,m}\log \tilde{p}_{n,m} + \log s \sum_{m=0}^n \tilde{p}_{n,m} m\\ &< -\sum_{m=0}^{m_0} \tilde{p}_{n,m}\log \tilde{p}_{n,m} -\sum_{m=m_0+1}^\infty 36\frac{s^2}{t^2} t^{4m} \big(\log( {\frac{36s^2}{t^2}}) + 4m\log t\big) + \log s \sum_{m=0}^\infty  36\frac{s^2}{t^2} t^{4m} m\\ &< \frac{m_0+1}{e} - 36\frac{s^2 t^{4m_0+2}}{1-t^4} \log( {\frac{36s^2}{t^2}}) - \frac{144 s^2 t^{4 m_0+2} (m_0(1- t^4) +1)}{(t^4-1)^2} \log t + \frac{36s^2 t^2}{(t^4-1)^2} \log s \\&\equiv C_0(s,t),
 \end{align*}
where we used $\sup_{x\in(0,1)}-x log(x)=e^{-1}$ for entropy terms with $m\leq m_{0}$ in the last inequality.
\end{proof}

Notice our proof here does not rely on the fact that $s>1$, and it applies to the $s=1$ case as well. 

\section{Scaling of the Spectral Gap} \label{SpGp}

\subsection{Super-exponential Upper bound in the $t>1,s=1$ Phase}

Since entanglement entropy is a measure of correlation in the system, a high entanglement entropy indicates that the system is highly correlated and also a gapless spectrum (in the thermodynamic limit) \cite{hastings2007area,arad2013area}. As our model at $t>1,s>1$ exhibits linear scaling of entanglement entropy, we expect the spectral gap to be also decreasing faster with system size than the $t=1$. Here, we give variational proof that the spectral gap for $t>1,s>1$ decreases exponentially with a square of the system size. 

Just as the linear scaling of entanglement entropy results from the prominence of the higher weighted paths in the ground state superposition, gaplessness can be shown by truncating lower weighted paths at the price of softly violating the superposition required to make the projectors in the Hamiltonian vanishing. 
To do so it is convenient to define a `prime walk' as follows: 
\begin{Def} A prime Dyck walk is a Dyck walk that is always above the x-axis, except at the endpoints. \end{Def}
By this definition, a Dyck walk is either prime or a concatenation of prime walks (Fig. \ref{excit} exhibits a Dyck walk in solid line made of two prime walks and one in dashed line made of three prime walks).

To construct a low energy variational excited state, we start with an auxiliary state that projects out all the walks in the ground state superposition whose longest prime walk has a length smaller than $n+1$. That is, define:
\begin{eqnarray}
P_{n,>}= \text{ \{s-colored walks containing a prime walk of length $l>n$ \} },
\end{eqnarray}
and $P_{n,<}=$ the complement of $P_{n,>}$ .
Our auxiliary state is defined as:
\begin{eqnarray}
 \frac{1}{\mathcal{N^*}} \sum_{\substack{w \in P_{n,>}} } t^{\frac{1}{2}\mathcal{A}(w)} |w\rangle.
\end{eqnarray}
For $t>1$ higher walks are favored rendering the auxiliary state largely overlapping with the ground state and therefore unqualified as a low energy excitation state. However, the color degree of freedom allows us to make this state orthogonal to the ground state by permuting the color of the last down move (or equivalently the first up move) in the longest prime walk. This way, all walks in the new superposition have one pair of spins with unmatched colors, and consequently orthogonal to all paths in the ground state. The choice of the `n+1' threshold on the cut-off of longest prime walk length eliminates the potential ambiguity in the location of the color permutation so that each path in the superposition has exactly one pair of unmatched colores. 

\begin{Thm} The spectral gap of the $t>1,s>1$ phase has an upper bound of $\frac{2(4s)^{n}}{1+t^2} t^{-n^2/2}$. \end{Thm} 

\begin{proof}

We define a new state $|\xi\rangle$ as:
\begin{equation} |\xi\rangle = \frac{1}{\mathcal{N^*}} \sum_{\substack{w \in P_{n,>}} } t^{\frac{1}{2}\mathcal{A}(w)}\ {\cal P}\ |w\rangle ,
\end{equation} 
where $\mathcal{N^*}$ is the new normalization factor and the operator ${\cal P}$ sends the color $c$ of the last down move of the longest prime walk to $c+1 \ mod \ s$ and leaves everything else unchanged. Because of the color imbalance we immediately have: \begin{equation} \langle \xi|GS\rangle =0, \end{equation} 
and $|\xi\rangle$ can be readily used as a variational wave function to bound the gap from above .

Let us compute the variational energy associated with the $|\xi\rangle$ state. First we note that:
\begin{eqnarray}
H_{\partial}|\xi\rangle=0,~~~~ H_X|\xi\rangle=0,
\end{eqnarray}
 as each non-matching color pair is separated by at least $n$ sites (while $H_X$ is only sensitive to nearest neighbor violations). The same goes for most of the projectors in $H_F$ just the way it works in the ground state. 

However, in $H_F$, we have also non-zero contributions coming from walks $w$ that are one "Fredkin" move away from leaving the set $P_n,>$. In other words, this happens when the first (second) projector in $H_F$ in Eq. \eqref{Ham} acts on the left (resp. right) endpoint of the longest prime walk and changes its length from $n+1$ to $n-1$ (the kind of which is absent in the superposition). For instance, applying the projectors on $\phi_{n-1,B}$ (Eq. \eqref{phiB}) to the prime walk $w$ corresponding to the one in Fig.~\ref{excit} gives:

\begin{figure}  \centering \includegraphics[width=0.8\textwidth]{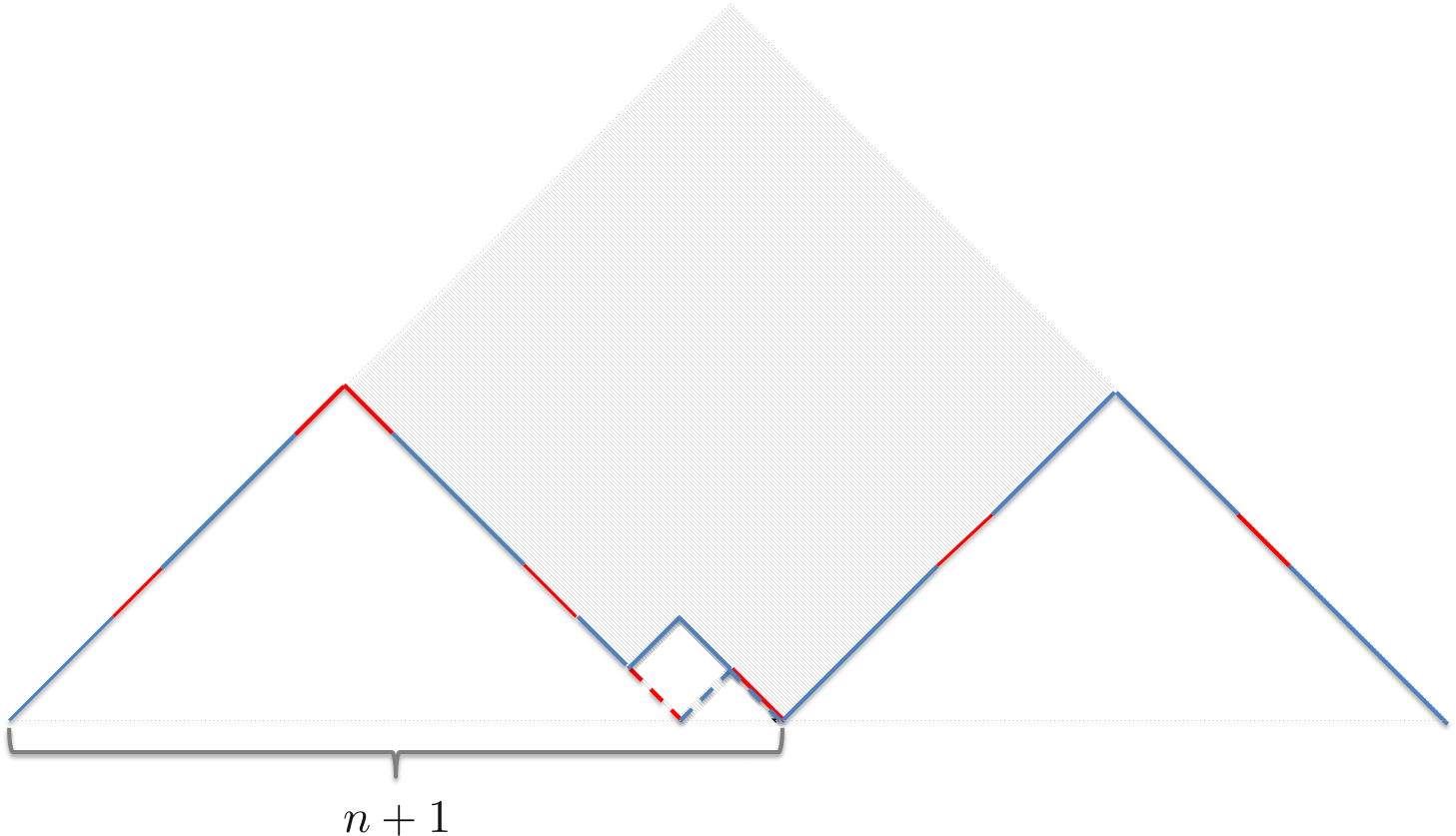} \caption{An representative walk in the superposition of $|\xi\rangle$ that crosses the threshold of the cut-off when acted on by the operator $(t|\downarrow^{r}_{n-1}\uparrow^{b}_{n}\downarrow^{b}_{n+1}\rangle - |\uparrow^{b}_{n-1}\downarrow^{b}_n\downarrow^{r}_{n+1}\rangle)(\langle\downarrow^{r}_{n-1}\uparrow^{b}_{n}\downarrow^{b}_{n+1}|t -  \langle\uparrow^{b}_{n-1}\downarrow^{b}_n\downarrow^r_{n+1}|)$.} \label{excit} \end{figure}

\begin{equation} \frac{1}{1+t^2}\langle w| (t|\downarrow^{r}_{n-1}\uparrow^{b}_{n}\downarrow^{b}_{n+1}\rangle - |\uparrow^{b}_{n-1}\downarrow^{b}_n\downarrow^{r}_{n+1}\rangle)(\langle\downarrow^{r}_{n-1}\uparrow^{b}_{n}\downarrow^{b}_{n+1}|t -  \langle\uparrow^{b}_{n-1}\downarrow^{b}_n\downarrow^r_{n+1}|)|w\rangle = \frac{1}{1+t^2}, \end{equation} and \begin{equation} \frac{1}{1+t^2}\langle w'|(t|\downarrow^{r}_{n-1}\uparrow^{b}_{n}\downarrow^{b}_{n+1}\rangle - |\uparrow^{b}_{n-1}\downarrow^{b}_n\downarrow^{r}_{n+1}\rangle)(\langle\downarrow^{r}_{n-1}\uparrow^{b}_{n}\downarrow^{b}_{n+1}|t -  \langle\uparrow^{b}_{n-1}\downarrow^{b}_n\downarrow^r_{n+1}|)|w\rangle = 0, \end{equation} 
with $w'$ is any other walk in the $|\xi\rangle$ (i.e. any other walk in $P_{n,>}$). 

We can now estimate the variational energy due to such paths. The number of these paths that will go from $P_{n,>}$ to $P_{n,<}$ when applying a Fredkin projector is very roughly bounded from above by $2^{2n} s^n$ (which is the total number of walks). On the other hand, the probability amplitudes of a path that has a prime walk length of exactly $n+1$ or $n+2$ in $P_{n,<}$, are penalized by their area differences from the highest weighted one, i.e. the shaded area in Fig.~\ref{excit}, by a factor smaller than $t^{-n^2/4}$. We therefore have the following upper bound:
\begin{equation} \langle \xi|H|\xi\rangle < \frac{2(4s)^{n}}{1+t^2} t^{-n^2/2}. 
\end{equation} 
Thus we have proved an upper bound of exponential of square of system size on the spectral gap when $t>1,s>1$.
\end{proof}
{\it Remark:} The overall factor $2$ above comes from possibility of modifying the prime path on the left or on the right.

\subsection{Exponential Upper Bound in the $t>1, s=1$ Phase}

As has been discussed in the previous subsection, a bounded from above entanglement entropy is expected to be a strong indicator of the existence of a non-vanishing spectral gap. Yet that intuition fails in the $t>1, s=1$ phase of the Motzkin spin chain. The numerical results in \cite{barbiero2017haldane} showed the $t>1, s=1$ Motzkin chain is gapless despite the boundedness of its entanglement entropy. Here we prove the Fredkin chain counterpart of this phenomenon, which can be readily adapted to the Motzkin chain.

We follow the same strategy we used to construct low energy excitation state from the $t>1, s>1$ phase, only now we don't have the luxury of taking advantage of color degrees of freedom to ensure the orthogonality to the ground state. Fortunately, there's still a degree of freedom we haven't fully exploited yet, namely the z-component of the total spin, or the net up spin of the chain, which can be non-vanishing when the boundary terms in the Hamiltonian is violated. To construct a low energy excitation due to this, we define \begin{equation*} Q_{-2}=\{\text{walks that starts from (0,0) and ends at (2n, -2) and never pass below x=-2.}\} \end{equation*}
Notice a Fredkin move acting on a walk in $Q_{-2}$ always gives another walk in $Q_{-2}$.

\begin{figure} \centering \includegraphics[width=0.8\textwidth]{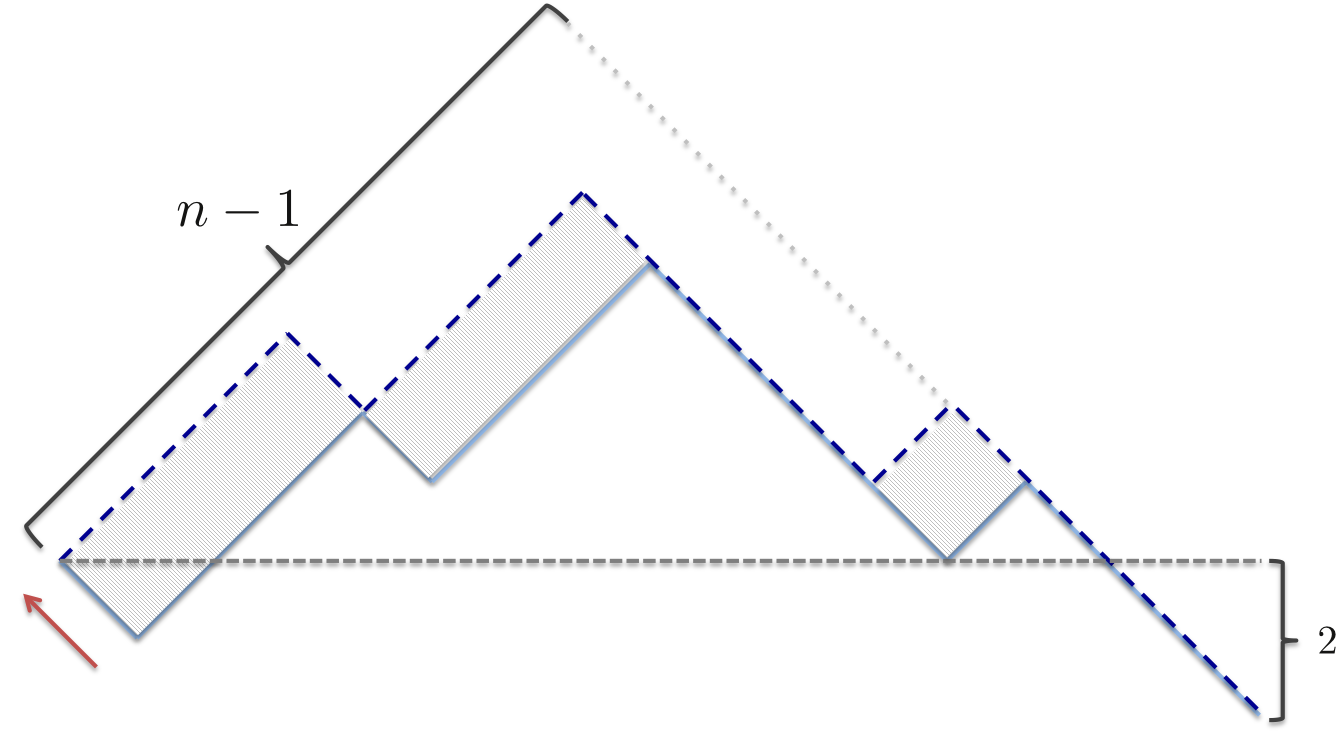} \caption{Two representative walks in $Q_{-2}$. The light blue one can be shifted in the direction of the orange arrow to become the dark blue one with a relative weight increase of $t^{n-1}$ corresponding to the area of the shaded regions.} \label{boundgap} \end{figure}

\begin{Thm} The spectral gap of the $t>1,s=1$ phase has an upper bound of $t^{-n+1}$. \end{Thm}

\begin{proof}
We define an excited state \begin{equation} |\eta \rangle = \frac{1}{N^\star} \sum_{w\in Q_{-2} } t^{\frac{1}{2}A(w)} |w\rangle, \end{equation} where $N^\star=(\sum_{w\in Q_{-2}} t^{A(w)})^{\frac{1}{2}}$ is the normalization factor. $|\eta\rangle$ is clearly orthogonal to the ground state as they have different total spins.
Since $|\eta\rangle$ only violates the boundary term in the Hamiltonian, after being acted on by $H$, only paths starting with a down move will survive. To get an estimate on the amplitude of the paths left, we point out that by rearranging the first down step to the last, (or equivalently shifting along the arrow in Fig.~\ref{boundgap},) we get another walk in $Q_{-2}$of area $2n-2$ bigger. Therefore, 
\begin{equation} \langle \eta|H|\eta\rangle = \frac{\sum_{w\in Q_{-2}}t^{A(w)}}{N^\star} < t^{-n+1}, \end{equation} 
which gives an upper bound on the spectral gap.
\end{proof}
{}
\section{Outlook} \label{Discussion}

We mention a few other future directions worth exploring. While we have shown how to construct a multi-parameter deformation, we have only studied entropy and gap for a uniform parameter $t$.  This choice keeps the chain translationally invariant, however, no momentum space arguments were involved in the analysis. A more general treatment will have to contend with the distribution of the $t_A$ parameters.

Second, the nature of the quantum phase transition is unclear. At a first glance, it hardly fits into the mechanism of symmetry breaking with an associated goldstone mode and exponents. To study the transition, as well as thermal effects, more detailed information about the density of states near the ground state is crucial. In particular, it would be very interesting to explore a possible field theoretic description in the continuous limit.

For the colored case our variational upper bound on the gap gives an elementary way of obtaining the gap behavior established in \cite{levine2016}. In \cite{levine2016} the colorful deformed Motzkin model was studied using more sophisticated mathematical machinery by utilizing the relation between frustration free local Hamiltonians and Markov chains, and applying a Cheeger inequality.

Using a different variational wavefunction, we have also proven that the spectrum is gapless for the colorless version of the model at $t>1$, in spite of entropy being bounded, furnishing an example of how bounded entanglement entropy does not imply a gap. 
The idea can be applied immediately to the deformed Motzkin chain providing an explanation for the surprising numerical observation of a vanishing gap in the $t>1,s=1$ phase \cite{barbiero2017haldane}. 

Finally, we remark that  \cite{barbiero2017haldane} also provides strong numerical evidence supporting the claim that the spectrum will be gapped when $0<t<1$ (for any $s$). It would be interesting to establish this observation rigorously.

{\bf Acknowledgments:} We would like to thank A. Ahamadain, R. Movassagh, V. Korepin and  H. Katsura for discussions.  The work was supported in part by the NSF grant DMR-1508245.


\end{document}